\documentclass[11pt]{article}
\usepackage{fullpage}
\usepackage[margin=1in]{geometry}
\usepackage[utf8]{inputenc}
\usepackage{hyperref}
\usepackage{natbib}
\usepackage{amsmath, xfrac}
\usepackage{amsfonts}
\usepackage{amssymb}
\usepackage{amsthm}
\usepackage[capitalize]{cleveref}
\usepackage{bbm}
\usepackage{enumerate}
\usepackage{mathtools}
\usepackage[dvipsnames]{xcolor}

\usepackage{tikz}
\usepackage{subfig}
\usepackage{graphicx}

\newtheorem{theorem}{Theorem}[section]
\newtheorem{lemma}[theorem]{Lemma}

\newtheorem{claim}[theorem]{Claim}

\newtheorem{definition}[theorem]{Definition}

\crefformat{scoringruleprogram}{#2program~(#1)#3}
\Crefformat{scoringruleprogram}{#2Program~(#1)#3}

\usepackage{ifthen}

\DeclareMathOperator{\argmax}{argmax}
\DeclareMathOperator{\argmin}{argmin}

\newcommand{\rbr}[1]{\left(#1\right)}

\newcommand{\lbr}[1]{\left\{#1\right\}}

\newcommand{\given}{\,|\,}

\newcommand{\prob}[2][]{\text{\bf Pr}\ifthenelse{\not\equal{}{#1}}{_{#1}}{}\!\left[{\def\givenn{\middle|}#2}\right]}
\newcommand{\expect}[2][]{\text{\bf E}\ifthenelse{\not\equal{}{#1}}{_{#1}}{}\!\left[{\def\givenn{\middle|}#2}\right]}

\newcommand{\expecto}[2][]{\text{\bf E}\ifthenelse{\not\equal{}{#1}}{^{#1}}{}\!\left[{\def\givenn{\middle|}#2}\right]}

\newcommand{\tparen}{\big}
\newcommand{\tprob}[2][]{\text{\bf Pr}\ifthenelse{\not\equal{}{#1}}{_{#1}}{}\tparen[{\def\given{\tparen|}#2}\tparen]}
\newcommand{\texpect}[2][]{\text{\bf E}\ifthenelse{\not\equal{}{#1}}{_{#1}}{}\tparen[{\def\given{\tparen|}#2}\tparen]}

\newcommand{\sprob}[2][]{\text{\bf Pr}\ifthenelse{\not\equal{}{#1}}{_{#1}}{}[#2]}
\newcommand{\sexpect}[2][]{\text{\bf E}\ifthenelse{\not\equal{}{#1}}{_{#1}}{}[#2]}

\newcommand{\abs}[1]{\left|#1\right|}
\newcommand{\indicate}[1]{{\bf 1}\left[#1\right]}

\newcommand{\cost}{c}

\newcommand{\score}{S}

\newcommand{\prior}{D}

\newcommand{\mos}{{\rm max-over-separate }}

\newcommand{\opt}{{\text{\rm IC-OPT}}}
\newcommand{\vopt}{{\text{\rm ALG-OPT}}}

\newcommand{\distoverposterior}{f}
\newcommand{\posterior}{\mu}

\newcommand{\outcome}{\omega}
\newcommand{\outcomes}{\Omega}

\newcommand{\val}{v}
\newcommand{\signal}{\sigma}
\newcommand{\signals}{\Sigma}

\newcommand{\threshold}{\eta}

\newcommand{\effortset}{\Psi}
\newcommand{\optset}{\effortset^*}

\newcommand{\event}{\mathcal{E}}

\newcommand{\strategy}{\tau}

\newcommand{\divergence}{\Lambda}
\newcommand{\KL}{{\rm KL}}

\newcommand{\subsetsum}{{\mathcal Z}}

\newcommand{\stopping}{\nparallel}

\newcommand{\groundset}{G}

\setcitestyle{authoryear,open={(},close={)}}
\title{Optimal Scoring Rules for Multi-dimensional Effort\thanks{The order of the authors are certified random. The records are available in \url{https://www.aeaweb.org/journals/policies/random-author-order/search?RandomAuthorsSearch\%5Bsearch\%5D=ZyHlOb3fJMgs}.}}
\author{Jason D. Hartline\thanks{Department of Computer Science, Northwestern University.
Email: \texttt{hartline@northwestern.edu}.}
\and Liren Shan\thanks{Department of Computer Science, Northwestern University.
Email: \texttt{lirenshan2023@u.northwestern.edu}.}
\and Yingkai Li\thanks{Cowles Foundation for Research in Economics, Yale University.
Email: \texttt{yingkai.li@yale.edu}.}
\and Yifan Wu\thanks{Department of Computer Science, Northwestern University.
Email: \texttt{yifan.wu@u.northwestern.edu}.}}
\date{}

\begin{document}

\maketitle
\begin{abstract}
  This paper develops a framework for the design of scoring rules to
  optimally incentivize an agent to exert a multi-dimensional effort.
  This framework is a generalization to strategic agents of the
  classical knapsack problem \citep*[cf.][]{BKV-05,sin-10} and it is
  foundational to applying algorithmic mechanism design to the
  classroom.  The paper identifies two simple families of scoring
  rules that guarantee constant approximations to the optimal scoring
  rule.  The truncated separate scoring rule is the sum of single
  dimensional scoring rules that is truncated to the bounded range of
  feasible scores.  The threshold scoring rule gives the maximum score
  if reports exceed a threshold and zero otherwise.  Approximate
  optimality of one or the other of these rules is similar to the
  bundling or selling separately result of \citet*{BILW-14}.
  Finally, we show that the approximate optimality of the best of those two simple scoring rules is robust when the agent's choice of effort is made sequentially. 
\end{abstract}


\section{Introduction}
\label{s:intro}

\reversemarginpar
\newcommand{\topic}[1]{\marginpar{\tiny #1}}
\renewcommand{\topic}[1]{{}}

\topic{Mechanism design for the classroom} This paper considers
mechanism design for the classroom.  An instructor aims to design a
grading mechanism that incentivizes learning, learning comes
from costly effort on the part of a student, and the student aims to
optimize their grade less the costs of effort.  Two key aspects of
this model for mechanism design are that effort is multi-dimensional
over a set of assigned tasks and that effort may lead to only partial
understanding of each task, i.e., effort does not generally guarantee
the student gets an answer correct.  The paper formulates this
problem as a multi-dimensional strategic version of the knapsack
problem and solves it by giving a simple and computationally efficient
scoring rule that incentivizes effort on an approximately optimal set
of tasks.

\topic{connections to agt} Strategic versions of the knapsack problem
and multi-dimensional mechanism design are of central interest in
algorithmic mechanism design.  For example, classic models describe
knapsack mechanisms for allocation \citep*[e.g.,][]{BKV-05} and for
procurement \citep*[e.g.,][]{sin-10}.  An important new frontier for
algorithmic mechanism design is in incentivizing private effort, e.g.,
to impact states as in contract theory
\citep*{dutting2022combinatorial}, or to collect information as in scoring
rules (this paper).  Optimization of scoring rules for
single-dimensional effort was considered by \citep*{HLSW-20}.  This
paper considers multi-dimensional effort where key steps in the
analysis resemble those of the well studied
bundling-or-selling-separately result of the multi-dimensional
mechanism design literature \citep*{BILW-14,BILW-20}.

\topic{Dialog between theory and practice in the classroom} Mechanism
design for the classroom has the potential to address a key challenge
for the two decade old field of algorithmic mechanism design.  To test
the theories of mechanism design in practice, the mechanisms must be
run in practice.  Unlike in classical algorithm design, where new
algorithms can be empirically evaluated on canonical data sets;
empirical validation of mechanisms fundamentally requires that their
inputs be from agents that are strategically responding to (other
agents and) the new mechanism.  Researchers of algorithmic mechanism
design do not generally have opportunities to test the classical
models of allocation or procurement.  Due to this challenge most
mechanisms of the algorithmic mechanism design literature have never
been empirically tested.  The classroom applications of mechanism
design, as proposed by this paper, provide immediate opportunities for
a dialogue between theory and practice; and their advances can lead to
better learning outcomes for students.  For example, \citet*{HLSW-20}
motivate their work on optimizing scoring rules for single-dimensional
effort by an empirical failure of the classical quadratic scoring rule
to provide sufficient incentives of effort for peer grading.

\topic{Knapsack Scoring} The {\em knapsack scoring problem} formulated
and solved in this paper is as follows.  There is a universe of tasks
that an instructor could assign to a student.  Effort of the student
on each task is binary.  Each task has a fixed learning value and a
fixed cost of effort.  The instructor aims to maximize the sum of
values of the tasks that the student puts effort on.  If effort were
directly observable, then this problem would be identical to the
knapsack problem: the optimal set of tasks to assign is the solution
to the knapsack problem with knapsack capacity equal to the maximum
grade and the student receives this maximum grade if effort is exerted
on all of the assigned tasks (zero otherwise).  Our instructor cannot
directly observe effort, but can instead administer a binary test for
each task where the student's belief about the answer to the test
improves with effort.  The instructor aims to select the set of tasks
that the student should perform and design a scoring rule with bounded
total score that incentivizes the student to perform these tasks.

\topic{Intuition from results} How does the instructor select the
tasks?  And how should the instructor score the student in aggregate?
The paper shows that there are two main cases that must be
considered.  Consider the case that scores from individual scoring
rules for the optimal set of tasks concentrate, e.g., because the
student is successful at many of them.  In this case then a good set
of tasks to incentivize can be found by greedily selecting tasks by
the ratio of value to cost and a {\em truncated separate scoring rule} can
incentivize effort on these tasks.  If the scores do not concentrate
then approximately optimal effort can be incentivized by the {\em threshold
scoring rule} and the tasks for this scoring rule can be identified by
greedily selecting tasks by the ratio of value to probability that the
student's effort is informative.
This observation is robust whether the agent exerts effort simultaneously or sequentially.


\paragraph{Related Work}


Prior work has considered mechanism design problems based on strategic
versions of the knapsack problem.  One framing is that of
single-minded multi-unit demand agents as buyers with a seller with a
multi-unit supply constraint.  In this model, only the values of the
agents can be strategically manipulated.  \citet*{BKV-05} considered
welfare maximization with this framing and gave a general method for
converting polynomial time approximation schemes (including the one
for knapsack) into incentive compatible mechanisms (with the same
approximation guarantees).  \citet*{AH-06} considers the same framing
with the goal of revenue maximization and a natural prior-free
benchmark.

Another knapsack framing reverses the buyer and seller roles: The
agents are sellers with private costs (object sizes in knapsack) and
the buyer aims to hire a team (set of sellers) to maximize value
but has a budget constraint (capacity of the knapsack).
\citet*{sin-10} posed this question and gave prior-free approximation
mechanisms when the buyers value function is submodular (generalizing
the linear value function of the traditional knapsack problem).
\citet*{BCGL-12,BCGL-17} considered the budget-feasibility question in
the Bayesian and prior-independent models of mechanism design and give constant approximations.
\citet*{BH-16} consider the Bayesian budget feasibility problem and
showed that posted pricing mechanisms give good approximation to the
Bayesian optimal mechanism.  In comparison to the literature on budget
feasibility, this paper's model of scoring rule optimization has a
single agent (resp.\ multiple agents) with a multi-dimensional strategy
space (resp.\ single-dimensional), the costs are public
(resp. private), but effort is private (resp.\ public).  With private
effort, the principal optimizing a scoring rule can only validate the
agent's effort in so far as the agent's posterior information from
effort improves over her prior information.

Multi-dimensional mechanism design problems are notoriously difficult.
In the classical setting of selling multiple items to a single agent
with multi-dimensional preferences, the algorithmic mechanism design
literature has identified simple constant-approximation mechanisms in a
number of canonical settings.  \citet*{BILW-14,BILW-20} show that for
an agent with independent additive values for multiple items then the
better of bundling or selling separately is a constant approximation.
\citet*{RW-15,RW-18} extend this approximation result to agents with
subadditive valuations.  See \citet*{BILW-20} for discussion of the
extensive literature generalizing these results.  These bundling
versus selling separately results are paralleled by this paper's
result showing that the better of truncated separate scoring or
threshold scoring is a constant approximation.

\cite{CB-16} consider a setting where a principal selects signal structures with knapsack constraints on the set of realizable signals. They show that when signals are substitutes, there exists a constant approximation algorithm for signal selection. However, in the general case, no algorithm can achieve a constant approximation with subexponentially many queries to the value of a signal. In our paper, we focus on the setting where the value function is submodular, which can be seen as a special case of substitutional signals if the agent's incentive is ignored. However, we prove that when the principal faces the task of designing an incentive scheme for the agent to select the set of signals, finding the optimal solution is NP-hard.


This work builds on the general framework for optimizing scoring rules
for effort that was initiated by \citet*{HLSW-20}.  Their main result
considers binary effort and multi-dimensional state.  In contrast, the
model of this paper is for multi-dimensional effort and
multi-dimensional state, but with a 1-to-1 correspondence between the
dimension of effort and state.


\citet*{CY-21} consider the design of scoring rules for maximizing a
binary effort in a max-min design framework.  For example,
complementing a prior-independent result from \citet*{HLSW-20}, they
show that the quadratic scoring rule is max-min optimal over a large
family of distributional settings.  \citet*{kon-22} apply the framework
of effort-maximization to multi-agent peer prediction where the
principal does not have access to the ground truth state and instead
must compare reports across several agents.

Several papers look at optimizing for multiple levels of a
single-dimensional effort with the objective of accuracy of the
forecast (i.e., the posterior from effort which is reported in a
proper scoring rule).  \citet*{osb-89} considers optimization of
quadratic scoring rules with a continuous level of
effort.  \citet*{Z-11} characterizes the optimal single-dimensional
scoring rule when the states are partially
verifiable.  \citet*{NNW-21} consider optimization of scoring rules for
integral levels of effort where the effort corresponds to a number of
costly samples drawn. \citet*{MB-22} characterizes the optimal scoring rule that maximizes revenue subject to a information cost,  with limited liability constraint. 

Optimization of effort in scoring rules has similarities to the
problem of optimizing effort in contracts, the main difference being
that, in the classical model of contract design the distribution over
states for each action is common knowledge.  In contract for scoring
rules, on taking an action the agent receives a signal that gives the
agent private information about the distribution of states. For the
contract design problems, \citet*{castiglioni2022designing} show that
the optimal contract can be computed in time polynomial in the number
of potential actions of the agent even when the costs of actions are
private information.  For the multi-dimensional effort model, the
number of actions is exponential in the size of the dimensions,
and \citet*{dutting2022combinatorial} show that with binary states, the
optimal contract can be computed in polynomial time if the function
mapping the action choices to the state distributions satisfies the
gross substitutes property, but is NP-hard when the function is more
generally submodular.

\paragraph{Future Directions}
The approach of the paper is one of Bayesian mechanism design where
the prior distribution is known to both the principal (instructor) and
agent (student).  Within the Bayesian model there are three main
directions for future work.  First, the positive results of this paper
are restricted to simplistic distributions over posteriors.  As
discussed in \Cref{sec:general info}, generalizing the results beyond this case
necessitates better upper bounds and richer families of approximation
mechanisms.  Second, our multi-dimensional effort-to-state mapping is
one-to-one.  It is an open direction to combine results for
multi-dimensional effort with the model of \citet*{HLSW-20} for
single-dimensional effort with multi-dimensional state.  Third, for
our motivating application in the classroom, the cost of effort varies
across students. It is an open direction to combine our model for
optimizing scoring rules with the model of budget feasibility where
the cost of effort is private.

Bayesian mechanism design is the first model in which to consider
novel mechanism design problems.  To obtain practical mechanisms,
however, it is important to consider robust versions of the problem.
The two canonical frameworks are that of prior-independence and sample
complexity.  Prior-independent framework looks to identify one
mechanism that has the best approximation to the Bayesian optimal
mechanism in worst case over distributions.  The sample complexity
framework looks to bound the number of samples necessary to obtain a
$1+\epsilon$ approximation to the Bayesian optimal mechanism.
\citet*{HLSW-20}, for example, gave such results for the problem of
designing scoring rules for a single-dimensional effort.  These are
open directions for optimizing multi-dimensional effort via scoring
rules.

\section{Preliminaries}
\label{sec:prelim}

This paper considers the problem of incentivizing effort from an agent to learn about an unknown state. 
There are $n$ tasks with state space $\outcomes = \times_{i=1}^n \outcomes_i$ where $\outcomes_i = \{0,1\}$.
For each task $i\in [n]$, 
state $\outcome_i\in\outcomes_i=\{0, 1\}$ is realized independently according to prior distribution $\prior$ which is the uniform distribution on $\outcomes_i$.
Exerting effort on task $i$ induces cost $\cost_i$ to the agent.
The agent can choose to exert effort on a set $\effortset\subseteq [n]$ of tasks at a cost $\sum_{i\in\effortset}\cost_i$.
Let $\signals$ be the signal space where $\bot\in\signals$ is an uninformative signal.
If the agent does not exert effort on task $i$, i.e.\ $i\notin \effortset$, 
with probability~$1$,
the agent receives an uninformative signal $\signal_i=\bot$ regardless of the realized state. 
If the agent exerts effort on task $i$, i.e.\ $i\in \effortset$, the agent receives a signal $\signal_i\in\signals$ according to a signal structure, 
 which is a random mapping from the states to the signal space.
Note that the signal structure on task $i$ induces a distribution 
$\distoverposterior_i$ over posterior $\posterior_i\in \Delta(\outcomes)$.

A special case that is of particular interest for our paper is when 
$\signals=\{0,1,\bot\}^n$
and the posterior belief is supported on $\{0,1,\sfrac{1}{2}\}^n$. 
In this case, if the agent exerts effort on task $i$, i.e.\ $i\in \effortset$,
with probability $p_i$,
the agent receives an informative signal $\signal_i = \outcome_i$,
and with probability $1-p_i$, the agent receives an uninformative signal $\signal_i=\bot$ regardless of the realized state. We call $p_i$ the state revelation probability of each task $i$. 
In the main body of the paper, we will focus on this special model, and discuss the extensions to general information structures in \cref{sec:general info}. 

Given the set of tasks $\effortset$ that the agent exerts effort on, 
the value of the principal is $\val(\effortset)$.
We assume that the valuation function $\val$ is \emph{submodular}: for every $\effortset'\subseteq \effortset\subseteq [n]$ of assignments, the principal's marginal value decreases, i.e.\ 
\begin{equation*}
    \forall i\in [n]\setminus \effortset,\, \val(\effortset'\cup \{i\})-\val(\effortset')\geq \val(\effortset\cup\{i\})-\val(\effortset).
\end{equation*} 
A special case of the submodular valuation is additive valuation, 
where $\val(\effortset) = \sum_{i\in\effortset} \val_i$
for given profile of $\{\val_i\}_{i\in[n]}$. 
The goal of the principal is to design a mechanism that maximizes her value subject to the budget constraint,
i.e., the payment to the agent is bounded between $0$ and $1$.
Note that if the effort choice of the agent can be observed by the principal, this problem reduces to the classical knapsack problem. 
The novel feature in our model is that effort is unobservable, and the principal can only score the agent according to the reported signals and realized states. 

\subsection{Static Effort Model}
In the static effort model, we assume that the agent makes the effort choice on all tasks simultaneously, 
and after the effort choice, 
the agent receives the signals on all tasks simultaneously.
By the revelation principle, 
it is without loss to restrict attention to mechanisms that 
recommend a set of tasks $\effortset$ for the agent to exert effort,
and after exerting effort,
incentivize the agent to truthfully report the received signal to the principal. 
Let $\prob[\sigma\sim\effortset]{\cdot}$ and $\expect[\sigma\sim\effortset]{\cdot}$ 
be the probability and expectation with respect to the distribution over signals conditional on exerting effort on set $\effortset$,
and let $\prob[\omega\sim\signal]{\cdot}$ and $\expect[\omega\sim\signal]{\cdot}$ 
be the probability and expectation with respect to the posterior belief of the agent conditional on receiving signal $\signal\in\signals$.
\begin{definition}\label{def:proper}
A scoring rule
$\score: \signals\times\outcomes\to [0,1]$ is \emph{proper} if 
for any $\signal,\signal'\in\signals$, 
\begin{align*}
    \expect[\omega\sim\signal]{\score(\signal,\outcome)} \geq \expect[\omega\sim\signal]{\score(\signal',\outcome)}.
\end{align*}
\end{definition}

Note that our definition of properness relies on the information structure and the set of signal realizations $\signals$. 
In principle, a scoring rule that satisfies our definition of properness may incentive the agent to misreport his belief that cannot be induced by those signal realizations. 
This may raise a concern for the robustness of the implemented scoring rule.
In \cref{apx:proper belief}, we show that it is without loss of generality to focus on scoring rules that are only proper for signal realizations in the support, 
by converting any such scoring rule to one that is 
proper for all possible beliefs without performance loss. 


\begin{definition}\label{def:ic}
A mechanism composed by a scoring rule
$\score: \signals\times\outcomes\to [0,1]$
and a corresponding recommendation set $\effortset$ is \emph{incentive compatible} 
if $\score$ is proper and 
for any $\effortset'\subseteq [n]$,\footnote{An alternative formulation of the mechanism is to only specify the scoring rule and delegate the computation of the optimal effort choice to the agent. 
However, the computation of the optimal effort choice may be NP-hard. 
The main advantage of our formulation is that we can ensure that the computation of the agent is simple. }
\begin{align*}
    \expect[\sigma\sim\effortset]{\expect[\omega\sim\signal]{\score(\signal,\outcome)}} - \sum_{i\in\effortset} \cost_i
    \geq \expect[\sigma\sim\effortset']{\expect[\omega\sim\signal]{\score(\signal,\outcome)}} - \sum_{i\in\effortset'} \cost_i.
\end{align*}
\end{definition}
\noindent The reward of the agent should be non-negative and the principal has a budget of $1$ for rewarding the agent. Thus, the score is ex-post bounded in $[0,1]$. 
Given the incentive constraints and reward constraints, the timeline of our model is as follows:
\begin{enumerate}
\item The principal commits to an incentive compatible mechanism with scoring rule
$\score: \signals\times\outcomes\to [0,1]$
and recommendation set $\effortset$.
\item The agent chooses a set $\bar{\effortset}$
of tasks on which to exert effort
and pays cost $\sum_{i\in \bar{\effortset}} \cost_i$.
\item States $\outcome=\{\outcome_i\}_{i=1}^n$ are realized, 
and the agent receives the signals $\signal\in\signals$.
\item The agent reports $\signal'$
and receives score $\score(\signal', \outcome)$.
\item The principal receives utility $\val(\bar{\effortset})$.
\end{enumerate}
Note that the agent is incentivized to choose $\bar{\effortset} = \effortset$ and truthfully reveal the signals in an incentive compatible mechanism. 
The \emph{knapsack scoring problem} for value function $\val$, costs $\{\cost_i\}_{i=1}^n$ and state revelation probabilities $\{p_i\}_{i=1}^n$ is formally defined as the following optimization program:
\begin{align*}
\opt(\val,\{\cost_i\}_{i=1}^n,\{p_i\}_{i=1}^n) =& \max_{\score,\effortset} \quad \val (\effortset)\\ 
&\quad\text{s.t.} \quad
(\score,\effortset) \text{ is incentive compatible for $\{\cost_i\}_{i=1}^n$ and $\{p_i\}_{i=1}^n$,} \\
&\qquad \quad\;\score(\signal,\outcome)\in[0,1], \quad \forall \signal,\outcome.\\
\intertext{We use the \emph{knapsack problem} for value function $\val$ and costs $\cost_i$ without incentive constraints as an upper bound on the knapsack scoring problem:}
\vopt(\val,\{\cost_i\}_{i=1}^n) =& \max_{\effortset\subseteq[n]} \quad \val (\effortset)\\ 
&\quad\text{s.t.} \quad \sum_{i\in \effortset} \cost_i \leq 1.
\end{align*}
It is easy to see that $\vopt(\val,\{\cost_i\}_{i=1}^n)\geq \opt(\val,\{\cost_i\}_{i=1}^n,\{p_i\}_{i=1}^n)$ for any $\val,\{\cost_i\}_{i=1}^n$ and $\{p_i\}_{i=1}^n$. 

The following characterization shows the budget-minimal scoring rule for incentivizing a single task. To minimize the budget, the agent is indifferent between: 1) reporting truthfully and non-truthfully; 2) exerting effort and not exerting the effort on the task. 
\begin{lemma}[\citealp*{HLSW-20}]
\label{lem:single opt}
With minimal budget $\frac{2\cost_i}{p_i}$, the agent can be incentivized to exert effort on a single task $\effortset = \{i\}$ with cost $\cost_i$ and probability $p_i$ of revealing. 
Moreover, the budget-minimal scoring rule for incentivizing effort is\footnote{By \cref{claim:proper belief}, $\frac{2\cost_i}{p_i}$ is also the minimum budget required for any scoring rule proper for belief elicitation
in order to incentivize the agent to exert effort on single task $\{i\}$.} 
\begin{align*}
\score_i(\signal_i,\outcome_i) = 
\begin{cases}
\frac{\cost_i}{p_i} & \signal = \bot\\
\frac{2\cost_i}{p_i}\indicate{\signal_i=\outcome_i} & \text{otherwise}.
\end{cases}
\end{align*}
\end{lemma}



By \Cref{lem:single opt}, with budget $1$, the agent can be incentivized to exert effort on a single task if and only if $\frac{2\cost_i}{p_i}\leq 1$.
\cref{lem:monotone task} shows that for multiple tasks there is a monotonicity property for the set of incentivizable tasks.
\begin{lemma}[Monotonicity in tasks]~\label{lem:monotone task}
For any set of assignments $\effortset\subseteq [n]$, 
if there exists a proper scoring rule~$\score$ 
such that the agent exerts effort on tasks $\effortset$, 
for any subset $\effortset'\subseteq\effortset$, 
there exists a proper scoring rule $\score'$ 
such that the agent exerts effort on tasks $i\in\effortset'$.
\end{lemma}
\begin{proof}
To incentivize effort on $\effortset'$, we construct $\score'$ by simulating the agent's effort on the set $\effortset\setminus\effortset'$. 
For any reported signal profile $\signal'$, 
let $\score'(\signal', \outcome)=\expect[\sigma\sim\effortset]{\expect[\omega\sim\signal]{\score(\signal,\outcome) \given \signal_i = \signal'_i, \forall i\in\effortset'}}$ be the score that ignores the report in set $\effortset\setminus\effortset'$, and takes expected score over this set by simulating the signals assuming effort.

The proof follows by showing that exerting effort on set $\effortset'$ is the optimal strategy for the agent with scoring rule $\score'$. 
Since the new scoring rule $\score'$ only depends on reported signals in set $\effortset'$, 
the agent has no incentive to exert effort on any task outside $\effortset'$. 
For any subset $\hat{\effortset}\subseteq\effortset'$, 
the expected utility difference between exerting effort in sets $\hat{\effortset}$ and $\effortset'$ given scoring rule $\score'$
is identical to 
the expected utility difference between exerting effort in sets $\hat{\effortset}\cup(\effortset\setminus\effortset')$ and $\effortset$ given scoring rule $\score$. 
Since exerting effort on all tasks in set $\effortset$ is optimal for scoring rule $\score$, 
exerting effort on all tasks in set $\effortset'$ is optimal for scoring rule $\score'$.
\end{proof}
By \cref{lem:single opt,lem:monotone task}, 
it is without loss to assume that $p_i\geq 2\cost_i$
for all tasks $i\in[n]$, 
and we will maintain this assumption throughout the paper. 

There are two families of scoring rules that will arise in our analysis,
\emph{truncated separate scoring rules}
and \emph{threshold scoring rules}.
Intuitively, 
the truncated separate scoring rules specify a scoring rule for each task,
and the total score is the sum of scores on each task, 
truncated by $0$ and the budget.

\begin{definition}\label{def:separate cap}
A scoring rule $\score$ is a \emph{truncated separate scoring rule} with budget $B > 0$
if there exists single-dimensional scoring rules $\score_1,\dots,\score_n$ and shifting parameter $d\geq 0$
such that 
$\score(\signal,\outcome) = \min\lbr{B, \max\lbr{0, -d+\sum_{i\in[n]} \score_i(\signal_i,\outcome_i)}}$.
\end{definition}

Note that due to the truncation to $[0,B]$, scoring rule $S$ may not be proper in general
even if the individual single-dimensional scoring rules are proper.
In later sections, we will properly design the parameter $d$ and single-dimensional scoring rules such that the aggregated scoring rule will remain proper. 



\begin{definition}\label{def:threshold}
A scoring rule $\score$ is a \emph{threshold scoring rule}
if there exist a recommendation set 
$\effortset\subseteq [n]$ and a threshold $ \threshold \geq 0$ on the number of tasks for the agent to predict correctly,
such that:
\begin{itemize}
    \item the score is $0$ if there exists task $i\in\effortset$ such that the reported signal is informative but wrong, i.e.\  $\signal_i\neq\bot$ and $\signal_i\neq\outcome_i$;
    \item let $k\triangleq\#\{i\in\effortset: \signal_i=\outcome_i\}$ be the number of tasks that the agent predicts correctly. 
    The score is $1$ if the agent's correct prediction exceeds the threshold, i.e.\ $k\geq \eta$; and $\frac{1}{2^{\threshold-k}}$ otherwise.
\end{itemize}
\end{definition}

The threshold scoring rule in \Cref{def:threshold} is proper. 
In \cref{apx:exp proper for belief}, to help with the understanding, we provide an equivalent formulation of threshold scoring rules in the special case of threshold $1$ such that it is also proper for eliciting the belief.
Here we show that it is also equivalent to the following non-proper scoring rule with the same recommendation set $\effortset$ and threshold $\threshold$:
\begin{itemize}
    \item the score is $1$ if both (1) the number of reported informative signal exceeds the threshold, i.e., $\#\{i\in\effortset: \signal_i\neq \bot\}\geq \threshold$;
    and (2) any task $i\in\effortset$ such that the reported signal is informative is correct, i.e., $\signal_i\neq\outcome_i$ if $\signal_i\neq\bot$;
    \item the score is $0$ otherwise.
\end{itemize}

Conditioning on the agent receiving $k\leq \threshold$ informative signals, his best response is to guess the rest $\threshold -k$ signals, with  a probability $\frac{1}{2^{\threshold -k}}$ that he can guess all correctly and receive score $1$. His expected utility is thus $\frac{1}{2^{\threshold -k}}$, which implies this non-proper scoring rule is equivalent to the proper scoring rule in \Cref{def:threshold}.


\subsection{Sequential Effort Model}

In the sequential effort model, we assume that the agent can sequentially exert effort on different tasks 
before the interaction with the seller,
and the agent can make effort decisions based on the signals he has received on previous tasks. 
Formally, at any moment, let $\hat{\effortset}$ be the set of tasks that the agent has exerted effort, 
and let $\signal_{\hat{\effortset}}$
be the set of signals on those tasks. 
The agent's strategy $\strategy(\hat{\effortset},\signal_{\hat{\effortset}})\in[n]\cup\{\stopping\}\backslash\hat{\effortset}$
specifies a new task to exert effort on or stop exerting more efforts (represented by $\stopping$)
based on history observations.
The timeline of our model is as follows:
\begin{enumerate}
\item The principal commits to a proper scoring rule
$\score: \signals\times\outcomes\to [0,1]$.
\item The agent adopts a sequential strategy $\strategy$ for exerting effort on tasks.
\item States $\outcome=\{\outcome_i\}_{i=1}^n$ are realized. 
The agent receives signals $\signal$ 
and pays cost $\sum_{i\in \bar{\effortset}} \cost_i$
where $\bar{\effortset}$ is the set of tasks that the agent has exerted effort on before stopping.
\item The agent reports $\signal$
and receives score $\score(\signal, \outcome)$.
\end{enumerate}
Let $\expecto[\strategy]{\cdot}$ be the expectation when the agent follows strategy $\strategy$ for exerting efforts. 

Note that for the sequential effort setting, we do not require the principal to make strategy recommendations to the agent. 
The main reason is because the agent's optimal search problem might be computationally hard
given the designed scoring rules. 
In this case, it would be unreasonable to prove the performance guarantee of our proposed scoring rules assuming that the agent can best responding to the mechanism. 
Instead, we make a weak assumption on agent's behavior, and show that for any reasonable response of the agent, 
the expected value of the set of tasks that the agent has exerted effort on is large enough. 

\begin{definition}\label{def:obvious dominated}
A strategy $\strategy$ is \emph{obviously dominated} if there exists $\hat{\effortset}$, signal $\signal_{\hat{\effortset}}$
and task $i\not\in\hat{\effortset}$ such that $\strategy(\hat{\effortset},\signal_{\effortset})=\stopping$ 
and the agent increases his expected utility by exerting effort on task $i$ compared to stopping, i.e., 
\begin{align*}
\expect[\sigma_i\sim\{i\}]{\expect[\omega\sim\signal_i\cup\signal_{\hat{\effortset}}]{\score(\signal,\outcome)}}
- \cost_i
\geq \expect[\omega\sim\signal_{\hat{\effortset}}]{\score(\signal,\outcome)}.
\end{align*}
\end{definition}

Requiring the agent's strategy to be not obviously dominated is in the same spirit of undominated strategies in \citet*{babaioff2009single} and advised strategies in \citet*{cai2020implementation},
with the adaption to sequential environments \citep*[c.f.,][]{li2017obviously}.
In \cref{sec:search}, we will show that the principal's payoff is approximately optimal given our proposed scoring rules
if the agent's strategy is not obviously dominated.

\section{Computational Hardness}
\label{sec:hardness}

In this section, we show that the design of the optimal mechanism for maximizing the principal's value 
is computationally hard by reduction from the NP-hard integer valued subset sum problem. 

\paragraph{Integer valued subset sum.}
Given $n$ integers $z_1,\dots,z_n$
and a target $Z > z_i$ for all $i\in[n]$, 
does there exists a set $\effortset\in[n]$
such that $\sum_{i\in\effortset}z_i = Z$?

Our proof idea is similar to the reduction from the subset sum to the knapsack problem.
The main challenge for reduction to our problem is that, 
in order to prevent the agent from randomly guessing the states of the tasks, 
there is a specific incentive constraint that determines the set of incentivizable tasks. 
The incentive compatible constraint potentially generates a much smaller value than the optimal set of tasks with total costs below the budget. 
To avoid this randomly guessing issue, 
we add additional tasks to the scoring rule design problem
such that the agent's utility from making any random guess
is sufficiently low, 
and that the optimal objective value of the principal exceeds a given value
if and only if the objective $Z$ of the subset sum problem can be achieved.

\begin{theorem}\label{thm:np-hard}
Computing the optimal mechanism in the knapsack scoring problem is NP-hard even if the valuation function is additive.
\end{theorem}
\begin{proof}
Given an integer valued subset sum instance
with integer parameters $z_1,\dots,z_n$ and $Z$, we construct a knapsack scoring problem. 
Let $\bar{\val} = 1+\sum_{i\in[n]}z_n$
and $\bar{\cost} = 1+\max_{i\in[n]}z_i$. 
Let $k$ be the minimum integer such that $2^{kn} > Z+2kn\bar{\cost}+1$. 
It is easy to see that the value of $k$ is polynomial in the number of digits to represent $Z$ and $\max_{i\in[n]}z_i$.
Construct a knapsack scoring problem with $(2k+1)n$ tasks
such that if the agent exerts effort on any task $i$, 
he observes the state $\outcome_i$ with probability $1$. The values and costs of the tasks are defined in the following way:
\begin{itemize}
    \item for each task $i\leq n$, let value and cost be $\val_i = \cost_i = z_i$;
    \item for each task $n+1\leq i\leq (2k+1)n$, let $\val_i = \bar{\val}$
and $\cost_i = \bar{\cost}$.
\end{itemize}
The budget of the principal is $Z+2kn\bar{\cost}+1$.
Note that this instance can be easily converted to our problem with budget 1 by re-scaling the budget and the costs by the same factor. 
We claim that the subset sum problem is true if and only if the optimal objective value for the knapsack scoring problem is $Z+2kn\bar{\val}$. 

If the optimal objective value for the knapsack scoring problem is $Z+2kn\bar{\val}$, 
this implies that in the optimal solution, 
the agent is incentivized to exert effort on all tasks $n+1\leq i\leq (2k+1)n$, 
which has a total contribution of $2kn\bar{\val}$. 
Thus the agent must exert effort on a subset $\effortset\subseteq[n]$ such that $\sum_{i\in\effortset} \val_i = Z$. 
Since $\val_i = z_i$ for all $i\in[n]$, 
$\effortset$ is a solution for the integer valued subset sum problem. 

If there exists a set of integers $\subsetsum\subseteq[n]$
such that $\sum_{i\in\subsetsum} z_i = Z$, 
consider the threshold scoring rule with recommendation set $\effortset=\subsetsum\cup\{n+1,\dots,(2k+1)n\}$ and threshold $\threshold=|\effortset|$, which scores budget $Z+2kn\bar{\cost}+1$ if the agents predicts all tasks in recommendation set $\effortset$ correctly. It is easy to verify that the utility of the agent for exerting effort on all tasks $i\in\effortset$ is $1$. The utility of the principal on recommendation set $\effortset$ is $Z+2kn\bar{\val}$. We are going to show this threshold scoring rule is incentive compatible and optimal. 

To prove this threshold scoring rule is incentive compatible, we divide agent's deviation into two cases:  1) the agent exerts effort on a small subset, so that he has to randomly guess  on a large number of tasks, which reduces his utility; 2) the agent exerts effort on a large subset, which induces a high total cost.
\begin{itemize}
    \item If the agent chooses to exert effort on a subset with size at most $\abs{\subsetsum} + kn$, he has to make random guess on at least $kn$ tasks. The utility of the agent is at most 
$2^{-kn}\cdot (Z+2kn\bar{\cost}+1) < 1$, 
which is strictly smaller than his utility for exerting effort on all tasks $i\in\effortset$.
\item If the agent chooses to exert effort on a subset with size between $\abs{\subsetsum} + kn$ and $\abs{\subsetsum} + 2kn -1$, 
the cost of effort for the agent is at least 
$Z+kn\bar{\cost} \geq \frac{1}{2}(Z+2kn\bar{\cost}+1)$ since $Z \geq 1$. 
Moreover, the expected payment to the agent 
is at most $\frac{1}{2}(Z+2kn\bar{\cost}+1)$
since the agent has to make a random guess on at least one task. 
This implies that the agent's utility is negative given this deviating strategy. 
\end{itemize}
Thus the agent's optimal choice is to exert effort on all tasks $i\in\effortset$.

Finally, we show that the optimal utility of the principal cannot exceed $Z+2kn\bar{\val}$.
Note that for the principal to obtain utility at least $Z+2kn\bar{\val}$, 
the agent must be incentivized to exert effort on all tasks $i\in \{n+1,\dots,(2k+1)n\}$
since the sum of value in $[n]$ is strictly below the value of any task $i\in \{n+1,\dots,(2k+1)n\}$. 
Moreover, the total cost of the agent for exerting effort given the optimal scoring rule is strictly less than $Z+2kn\bar{\cost}+1$
since the agent can obtain strictly positive utility by exerting no effort and randomly guessing. 
Since the costs are integer valued, 
the total cost is at most $Z+2kn\bar{\cost}$.
As the total cost for exerting effort on tasks $i\in \{n+1,\dots,(2k+1)n\}$
is $2kn\bar{\cost}$, 
the cost of the agent on tasks within subset $[n]$ is at most $Z$. 
Since the value coincides with the cost in this case, 
the value of the principal from incentivizing the agent to exert effort on tasks within $[n]$ is at most $Z$. 
Therefore, the optimal utility of the principal is $Z+2kn\bar{\val}$.
\end{proof}

\section{Bicriteria Approximation: Inflating the Budget}
\label{sec:budget inflation}


In this section, we show that there exists a proper 
truncated separate scoring rule with a constant budget 
that achieves higher value for the principal
than the optimal mechanism with budget $1$. 
Specifically, we show that by inflating the budget by a constant factor, 
the principal is able to attain at least the optimal objective value given budget~$1$
with relaxed incentive constraints. 

The approximation mechanism we design for the knapsack scoring problem uses the (approximately) optimal solution for the knapsack problem as a blackbox. 
Note that for general submodular valuations, computing the optimal solution for $\vopt$ is NP-hard. 
The following lemma shows that there exists a polynomial time algorithm to get an $\sfrac{e}{(e-1)}$-approximation.
\begin{lemma}[\citealp*{sviridenko2004note}]
\label{lem:submodular greedy}
For submodular valuation $\val$, there exists a polynomial time algorithm that computes a feasible solution $\effortset$
such that $\val(\effortset)\geq (1-\sfrac{1}{e})\vopt$.
\end{lemma}

\begin{figure}[t]
    \centering
    \fbox{
    \parbox{0.96\textwidth}{

    \textbf{Truncated Scoring Mechanism} for additive values with budget $11$
    
    Post the truncated scoring rule on a recommendation set $\effortset$
    
        \begin{itemize}
            \item For each assignment $i\in \effortset$, let the budget-minimal scoring rule be $\hat{\score}_i$.
            
            Posting single dimensional scoring rules:
                \begin{align*}
\score_i(\signal_i,\outcome_i)=\frac{9}{8}\hat{\score}_i(\signal_i,\outcome_i) = 
\begin{cases}
\frac{9\cost_i}{8p_i} & \signal = \bot\\
\frac{9\cost_i}{4p_i}\cdot \indicate{\signal_i=\outcome_i}  & \text{otherwise}
\end{cases}
\end{align*}
            \item Sum over the single dimensional scores, and truncate back to $[0, 11]$:
            \begin{align*}
                \score=\max\lbr{0, \min\lbr{11, \sum_i\score_i-d}},
            \end{align*}
            where $d=-\frac{11}{2}+\frac{9}{8}\sum_{i\in\effortset} \frac{\cost_i}{p_i}$ is the shift on the sum. 
            \end{itemize}

    }
    }
    \caption{Truncated Scoring Mechanism.}
    \label{fig:truncated scoring mechanism}
\end{figure}

\begin{figure}[t]
     \centering
    \fbox{
    \parbox{0.96\textwidth}{

\textbf{Recommendation set $\effortset$} for truncated scoring mechanism
        
        Input: ground set $\groundset$
        
        Output: set $\effortset$
        
        Greedily include tasks from $\groundset$ to $\effortset$, by value-cost ratio with a budget $\frac{3}{2}$ on the total cost. 
    }
    }
    \caption{Procedure for identifying optimal recommendation set for truncated scoring mechanism.}
    \label{fig:truncated scoring set}
\end{figure}
\begin{theorem}\label{thm:budget inflation}
The truncated scoring mechanism in \Cref{fig:truncated scoring mechanism} with a budget $B=11$
 guarantees value at least 
the optimal knapsack value (\vopt). 
Moreover,  for submodular values, 
there is a polynomial time algorithm for computing the recommendation set $\effortset$, which attains an $\sfrac{e}{(e-1)}$-approximation.
\end{theorem}

The main idea is that with multiple tasks, the sum of the scores on different tasks concentrates
around its expectation. 
Therefore, we can take the sum of the scores and shift it such that the expected score of not exerting any effort is only one half of the budget $11$. 
Moreover, with an inflated budget,
we can ensure that the ex post shifted sum remains in the range of $[0,11]$ with high probability,
and hence the agent's incentive is almost aligned with his incentive in separate scoring rules without the truncation. 
This allows us to show that the designed truncated separate scoring rule is proper, 
and the agent has the incentive to follow the recommendation. 
The detailed proof of the theorem is provided in \cref{apx:proof budget inflation}.

\section{Value Approximation}
\label{sec:approximation}

In this section, we show that the better of 
a truncated separate scoring rule
and a threshold scoring rule 
is a constant approximation to the optimal value of the knapsack scoring problem (\opt). 
The idea is to divide the set of tasks into two subsets based on whether the sum of
optimal individual single-dimensional scoring rule concentrates, 
and then design approximately optimal scoring rule for each subset separately. 
This concept is analogous to the core-tail decomposition adopted for multi-item auctions \citep*{BILW-20}, 
while the details for proving the results are quite different.

The first case is to consider tasks such that their costs are small 
compared to their probabilities of revealing the state when the agent exerts effort. 
In this case, the budget required for incentivizing each single task is small.
Thus, analogous to \cref{thm:budget inflation},
the variance of the score for incentivizing each task separately is small
and the sum of the scores concentrates well given the total budget 1. 
This implies that the ex post sum is close to its expectation with high probability.
By truncating the sum of optimal single-dimensional scoring rules to comply with the ex post budget constraint, 
the incentives of the agent for exerting effort is barely affected, 
and we obtain a constant approximation to the knapsack solution in this case. 

The second case is to consider tasks such that their costs are large 
compared to their probabilities of revealing the states when the agent exerts effort.
Unlike the traditional knapsack problem where large costs on the tasks indicate the existence of a single task with valuation close to the optimal, 
in the effort incentivization problem, 
there still exists the hard case where in the optimal mechanism, the agent need to be incentivized to exert effort on a large number of tasks and each task only contributes to a small fraction of the optimal objective value. 
Moreover, since the probabilities of revealing the states are small, 
the expected number of tasks on which the agent receives informative signals is small 
and hence the sum of scores may not concentrate. 
Alternatively, we show that in this case, 
the score of the agent has to be close to the budget if 
he receives an informative signal on any task. 
Therefore, to incentivize the agent to exert effort on any task $i$, 
the total probability that the agent gets an informative signal on any task $i'\neq i$ cannot be too large
because otherwise the principal will not have enough budget to incentivize task $i$ after rewarding the agent for acquiring an informative signal on task $i'$. 
Thus, an upper bound is imposed on the sum of probabilities
for the set of incentivizable tasks. 
We find a set of tasks that can be incentivized by a threshold scoring rule through a greedy algorithm on the ratio of the value to the probability,
and show that the value of this set is a constant approximation to the value given by the optimal scoring rule. 
\begin{figure}[t]
    \centering
    \fbox{
    \parbox{0.96\textwidth}{

    \textbf{Threshold Scoring Mechanism} for additive values with budget $1$

    Post the threshold scoring rule on a recommendation set $\effortset$
        \begin{itemize}
            \item Score $1$ if both (1) at least $1$ reported signal in $\effortset$ is informative; and (2) any task reported signal that is informative is correct.
            \item Score $0$ otherwise.
        \end{itemize}
    }
    }
    \caption{Threshold Scoring Mechanism.}
    \label{fig:threshold scoring mechanism}
\end{figure}
\begin{figure}[t]
    \centering
    \fbox{
    \parbox{0.96\textwidth}{

    \textbf{Recommendation set $\effortset$} for threshold scoring mechanism
    
    Input: ground set $\groundset$.
    
        For each task $j$ in the ground set $\groundset$:
        \begin{itemize}
            \item  initialize by adding $j$ into the recommendation $\effortset^j=\{j\}$;
            \item update the ground set $\groundset$:
            $\groundset^j=\{i\in \groundset\mid 1-\frac{2\cost_j}{p_j}+p_j\leq 1-\frac{2\cost_j}{p_i}+p_j\}$;
            
            \item greedily include tasks from $\groundset^j$ by the value-probability ratio $\frac{\val_i}{p_i}$ with a  budget\\
            $\sum_{j\in \effortset^j}p_i\leq 1-\frac{2\cost_j}{p_j}+p_j$;
            \item Consider set $\effortset'^j=\{j, j^*\}$, where $j^*=\argmax_{i\in \groundset^j}\val(i)$ is the most valuable task.
            
            Take the better of the knapsack solution and the set $\effortset'^j$
            
        \end{itemize}
        
        Output the set with maximum value: $\effortset=\argmax_{\effortset^j}\val(\effortset^j)$.
    }
    }
    \caption{Procedure for identifying approximately optimal recommendation set.}
    \label{fig:recommendation set static}
\end{figure}

\begin{theorem}\label{thm:approximation}
The better of 
a truncated separate scoring rule
and a threshold scoring rule 
is a $1091$-approximation to the optimal value of the knapsack scoring problem (\opt).
Moreover, for additive values, the parameters of such mechanism can be computed in polynomial time, 
and for submodular values, 
there is a polynomial time algorithm for computing the parameters that loses an additional multiplicative factor of $\sfrac{e}{(e-1)}$ in approximation ratio.
\end{theorem}

We first show an upper bound on the sum of state revelation probabilities for each set of incentivizable tasks 
when the ratio of the cost to the probability for any task in this set is large.

\begin{lemma}\label{lem:asym-upper-bound}
For any set $\effortset\subseteq[n]$ such that
$p_i\leq \frac{1}{4}$ and
$\frac{2\cost_i}{p_i}\geq \frac{15}{16}$ for all tasks $i\in\effortset$, 
if the set $\effortset$ can be incentivized by a proper scoring rule with budget $1$, there exists a budget-pivotal task $i^*=\argmin_{i\in\effortset}\frac{16}{3}\rbr{1-\frac{2\cost_{i^*}}{p_{i^*}}}+p_{i^*}$, such that the budget over total revealing probabilities is determined by $i^*$:
$$\sum_{i\in \effortset} p_i \leq \frac{16}{3}\rbr{1-\frac{2\cost_{i^*}}{p_{i^*}}}+p_{i^*}.$$
\end{lemma}
\begin{proof}
We first define several useful notations. We define $\event$ to be the event that the agent receives no informative signal on all tasks in $\effortset$. 
Let $q_0= \prob{\event}=\Pi_{j\in \effortset} (1-p_j)$ be the probability that event~$\event$ happens. 
Let $s_0 = \expect[\outcome\sim\signal]{\score(\signal, \outcome)\given \event}$
be the expected score of the agent when he receives no informative signal. We also define $\event_i$ to be the event that the agent receives no informative signal on all tasks in $\effortset \backslash \{i\}$. let $q_i=\prob{\event_i} = \Pi_{i\in \effortset\setminus\{i\}}(1-p_j)$ be the probability that the event $\event_i$ happens.
Let $s_i=\expect[\outcome\sim\signal]{\score(\signal, \outcome)\given \event_i, \signal_i \neq \bot}$ be the expected score of the agent when he only receives an informative signal on task~$i$. 

Next we divide the analysis into two cases: (1) $q_0 \geq \sfrac{1}{2}$; and (2) $q_0 < \sfrac{1}{2}$. 
\begin{enumerate}[{Case} 1:]
\item $q_0 \geq \sfrac{1}{2}$. 
In this case, we first show that the expected score for no informative signal $s_0$ can not be less than $\sfrac{1}{4}$. Suppose $s_0 < \sfrac{1}{4}$, then we show that the incentive constraint for exerting effort on any task $i$ is violated.
The utility increase of the agent for exerting effort on task $i$ is 
\begin{align*}
&\expect[\signal\sim\effortset]{\expect[\outcome\sim \signal]{\score(\signal,\outcome)}} - \expect[\signal\sim\effortset\setminus\{i\}]{\expect[\outcome\sim\signal]{\score(\signal,\outcome)}} \\
&= p_{i} \rbr{\expect[\signal\sim\effortset]{\expect[\outcome\sim\signal]{\score(\signal,\outcome)\given \signal_{i}\neq\bot}} 
- \expect[\signal\sim\effortset]{\expect[\outcome\sim\signal]{S(\signal,\outcome) \given \signal_{i} = \bot}}}
\end{align*}
Then, we bound the expected score increase for receiving an informative signal on task~$i$. 
Conditioned on event $\event_i$, the expected score difference is $s_i - s_0$. Since the scoring rule is proper, 
we have $s_0 \geq s_i/2$, which implies $s_i - s_0 \leq s_0 \leq \sfrac{1}{4}$.
Conditioned on the complement event $\bar{\event}_i$, by the properness of scoring rule, the expected score difference is at most $\sfrac{1}{2}$.
Thus, the utility increase for exerting effort on task $i$ is at most
$$
\expect[\signal\sim\effortset]{\expect[\outcome\sim \signal]{\score(\signal,\outcome)}} - \expect[\signal\sim\effortset\setminus\{i\}]{\expect[\outcome\sim\signal]{\score(\signal,\outcome)}} \leq p_{i} \rbr{q_{i} (s_i-s_0) + \frac{1}{2}(1-q_{i})} 
< \frac{3p_{i}}{8} < \cost_{i},
$$
which violates the incentive constraint for exerting effort on task $i$.

Therefore, we have $s_0\geq \sfrac{1}{4}$. We now lower bound the expected score $s_i$ for receiving only one informative signal on task $i$.
For any task $i$, the incentive constraint implies that 
\begin{align*}
\cost_{i}&\leq 
\expect[\signal\sim\effortset]{\expect[\outcome\sim\signal]{S(\signal,\outcome)}} - \expect[\effortset\backslash\{i\}]{\expect[\signal]{S(\signal,\outcome)}} \\
&= p_{i} \rbr{\expect[\signal\sim\effortset]{\expect[\outcome\sim\signal]{S(\signal,\outcome)\given \signal_{i}\neq\bot}} 
- \expect[\signal\sim\effortset]{\expect[\outcome\sim\signal]{S(\signal,\outcome) \given \signal_{i} = \bot}}} \\
&\leq p_{i} \rbr{q_{i} (s_{i}-s_0) + \frac{1}{2} (1-q_{i})}.
\end{align*}
Since $q_i \geq q_0 \geq \sfrac{1}{2}$ and $\sfrac{c_i}{p_i}\geq \sfrac{15}{32}$, this further implies that 
\begin{align*}
s_{i} \geq
s_0 + \frac{\frac{\cost_{i}}{p_{i}}-\frac{1}{2}(1-q_{i})}{q_{i}} \geq \frac{11}{16}.
\end{align*}

Consider any fixed task $i^*\in \effortset$. Let $\hat{s} = \expect[\signal\sim\effortset]{\expect[\omega\sim\signal]{S(\signal,\omega)\given \signal_{i^*} = \bot, \bar{\event}_{i^*}}}$ be the expected score of the agent when he has no signal on task $i^*$, 
and at least one informative signal on tasks in $\effortset\backslash\lbr{i^*}$.
Since the scoring rule is proper, $\hat{s} \geq \min_{i}s_{i} \geq \sfrac{11}{16}$. 
The incentive constraint on task $i^*$ implies that 
\begin{align*}
\cost_{i^*}&
\leq p_{i^*} \rbr{\expect[\signal\sim\effortset]{\expect[\outcome\sim\signal]{S(\signal,\outcome)\given \signal_{i^*}\neq\bot}} 
- \expect[\signal\sim\effortset]{\expect[\outcome\sim\signal]{S(\signal,\outcome) \given \signal_{i^*} = \bot}}} \\
&\leq p_{i^*} \rbr{\frac{q_{i^*}}{2} + (1-q_{i^*})(1-\hat{s})},
\end{align*}
where the last inequality is due to the expected score difference conditioned on $\bar{\event}_i$ is at most $1-\hat{s}$.
Hence, we have 
\begin{align*}
q_{i^*} \geq 1-\frac{8}{3}\rbr{1-\frac{2\cost_{i^*}}{p_{i^*}}}.
\end{align*}
Note that the probability that the agent receives at least one informative signal in $\effortset\backslash\lbr{i^*}$
is at least the sum of probability that 
the agent receives an informative signal on task $i$
and zero informative signal on tasks in $\effortset\backslash\lbr{i^*,i}$.
Note that the probability of the latter event is at least $q_0\geq \sfrac{1}{2}$.
Thus, it holds that
\begin{align*}
1-q_{i^*} \geq \frac{1}{2}\sum_{i\in\effortset\backslash\{i^*\}} p_i.
\end{align*}
By combining the two inequalities above, we have 
\begin{align*}
\sum_{i\in \effortset\backslash\lbr{i^*}} p_i \leq 2(1-q_{i^*}) 
\leq \frac{16}{3}\rbr{1-\frac{2\cost_{i^*}}{p_{i^*}}}.
\end{align*}

\item Suppose $q_0 < \sfrac{1}{2}$. Consider any fixed task $i^*\in \effortset$. In this case, we first show that there exists a subset $\bar{\effortset}\subseteq\effortset$
which satisfies the following three properties: (1) $i^*\in\bar{\effortset}$;
(2) $\bar{\effortset}$ can be incentivized by a proper scoring rule;
and (3) the probability of no informative signal in $\bar{\effortset}\backslash \{i^*\}$ is between $[\sfrac{1}{2},\sfrac{2}{3})$. By case $1$, this subset $\bar{\effortset}$ cannot be incentivized, which is a contradiction. 

To find such a subset, we remove tasks in $\effortset \backslash \{i^*\}$ from $\effortset$ one by one randomly. Since $p_{i^*} \leq \sfrac{1}{4}$ and $q_0 < \sfrac{1}{2}$, we have $q_{i^*} = q_0/(1-p_{i^*}) < \sfrac{2}{3}$. If $q_{i^*} \in [\sfrac{1}{2},\sfrac{2}{3})$, then $\effortset$ satisfies three properties. We use $\effortset'$ to denote the subset in this deletion process. Let $q_{i^*}'$ be the probability of no informative signal in $\effortset' \backslash\{i^*\}$. If $q_{i^*} < \sfrac{1}{2}$, then we have $q_{i^*}'$ increases from $q_{i^*}$ to $1$ during this process. If there is no $q_{i^*}' \in [\sfrac{1}{2},\sfrac{2}{3})$ in this process, then there exists a task $i \in \effortset$ with $p_i > \sfrac{1}{4}$, which contradicts the assumption. Let $\bar{\effortset}$ be the subset with probability $\bar{q}_{i^*} \in [\sfrac{1}{2},\sfrac{2}{3})$ during this process. It is easy to see that $\bar{\effortset}$ satisfies other two properties.

However, by union bound, 
\begin{align*}
\sum_{i\in\bar{\effortset}\backslash\{i^*\}} p_i \geq 1-\bar{q}_{i^*} > \frac{1}{3} 
\geq \frac{16}{3}\rbr{1-\frac{2\cost_{i^*}}{p_{i^*}}},
\end{align*}
which contradicts the assumption that $\bar{\effortset}$ can be incentivized according to the case 1. \qedhere
\end{enumerate}
\end{proof}

\begin{proof}[Proof of \cref{thm:approximation}]
We first prove the theorem for additive valuations, 
and then at the end we introduce the details for generalizing our techniques to submodular valuations. 
Recall that for any task $i$, we have $p_i\geq 2\cost_i$ since otherwise that task cannot be incentivized by the principal. 
Thus, we divide the tasks into two sets $X, Y$ based on the ratio $\sfrac{p}{2\cost_i}$ as follows
$$
X = \lbr{i: \frac{p_i}{2\cost_i} > 11}; 
\qquad Y =\lbr{i: 1 \leq \frac{p_i}{2\cost_i} \leq 11}.
$$
By \Cref{thm:budget inflation}, there is a truncated separate scoring rule with budget $1$ that is an $11$-approximation on the set~$X$ 
since this case can be viewed the same as the one in \cref{thm:budget inflation}
by scaling the score and the costs by the same constant factor~$11$. 

We divide the set $Y$ into three subsets. 
\begin{align*}
Y_1 = \lbr{i: p_i\geq \frac{1}{4},
1 \leq \frac{p_i}{2\cost_i} \leq \frac{16}{15}}; \quad
Y_2 =\lbr{i: p_i< \frac{1}{4}, 
1 \leq \frac{p_i}{2\cost_i} \leq \frac{16}{15}};\quad
Y_3 =\lbr{i: \frac{16}{15} < \frac{p_i}{2\cost_i} \leq 11}.
\end{align*}
Intuitively, set $Y_1$ corresponds to the case that the costs of effort are large, and it is sufficient to only incentivize one task with highest value in this set. 
Both set $Y_2$ and $Y_3$ corresponds to the situation where 
the probabilities of revealing the states are small compared to the costs, and hence the concentration technique cannot be applied.
In both cases, we utilize \cref{lem:asym-upper-bound} to bound the sum of probabilities for any set of incentivizable tasks, 
and hence showing that the set of tasks we identified by our polynomial time algorithm is approximately optimal.

\begin{enumerate}[{Case} 1:]
\item $Y_1 = \lbr{i: p_i\geq \frac{1}{4},
1 \leq \frac{p_i}{2\cost_i} \leq \frac{16}{15}}$. 
In this case, 
$\cost_i\geq \frac{15p_i}{32} \geq \frac{15}{128}$.
Therefore, at most $8$ tasks in $Y_1$ can be incentivized simultaneously in the optimal mechanism. 
By choosing the task in $Y_1$ with highest value, 
the principal attains an $8$-approximation by only incentivizing that task.

\item $Y_2 =\lbr{i: p_i< \frac{1}{4}, 
1 \leq \frac{p_i}{2\cost_i} \leq \frac{16}{15}}$. We use the threshold mechanism in \Cref{fig:threshold scoring mechanism}, with a recommendation set generated by running \Cref{fig:recommendation set static} on set $Y_2$.

We prove it is a $\frac{32}{3}$-approximation by showing: (1) the threshold scoring mechanism is incentive compatible (i.e.\ the agent's best response is to exert effort on all tasks in the recommendation set); and (2) the total value in the recommendation set $\effortset$ is a $16$-approximation of the optimal solution. 

\begin{enumerate}
    \item[(1)] The threshold scoring mechanism is incentive compatible. Specifically, we show that the set $\effortset_j$ can be incentivized for any task $j \in Y_2$.
For any task $j \in Y_2$, and any $i'\neq j, i'\in \effortset_{j}$, according to two constraints used in the construction of $\effortset_{j}$,
we have 
\begin{align*}
\sum_{i\in \effortset_{j}\backslash\{i'\}} p_i
= \sum_{i\in \effortset_{j}\backslash\{j\}} p_i 
- p_{i'} + p_{j}
\leq \rbr{1-\frac{2\cost_{i'}}{p_{i'}}}.
\end{align*}
Given the threshold scoring rule with threshold $\eta=1$ on effort set $\effortset_{j}$,
the expected score increase of exerting effort on task $i'$ is at least the probability of receiving no informative signal on tasks in $\effortset_{j}\backslash\{i'\}$
times the conditional score increase for exerting effort. By the union bound, we have the probability of receiving no informative signal on tasks in $\effortset_{j}\backslash\{i'\}$ is at least $\Pi_{i\in \effortset_{j}\backslash\{i'\}} (1-p_i) \geq 1-\sum_{i\in \effortset_{j}\backslash\{i'\}} p_i$. Conditioned on this event, the expected score increase for exerting effort on $i'$ is $p_{i'} + p_{i'}/2 - 1/2 = p_{i'}/2$. Thus, we have the expected score increase of exerting effort on task $i'$ is at least
\begin{align*}
\rbr{1-\sum_{i\in \effortset_{j}\backslash\{i'\}} p_i}\cdot \frac{p_{i'}}{2}
\geq \cost_{i'}.
\end{align*}
Therefore, for all searches $j \in Y_2$, a threshold scoring rule with threshold $1$ and recommendation set $\effortset_{j}$ is incentive compatible.

\item[(2)] The total value in the recommendation set $\effortset$ is a $16$-approximation of the optimal solution. 
By \cref{lem:asym-upper-bound}, 
for any set $\effortset'\subseteq Y_2$ that can be incentivized, and any $i^*\in\effortset'$, 
we have
\begin{align*}
\sum_{i\in \effortset'\backslash\lbr{i^*}} p_i \leq 
\frac{16}{3} \rbr{1-\frac{2\cost_{i^*}}{p_{i^*}}}.
\end{align*}

Let $\optset$ be the optimal effort set in the knapsack scoring problem when the set of available tasks is $Y_2$. 
Let $\hat{i}=\argmin_{i\in \optset}\rbr{ 1-\frac{2\cost_{i}}{p_{i}}+p_i}$ be the budget-pivotal task. This can be interpreted as a budget over the total probabilities in the optimal set $\optset$:
\begin{align*}
    \sum_{i\in \optset} p_i \leq
\frac{16}{3} \rbr{ 1-\frac{2\cost_{\hat{i}}}{p_{\hat{i}}}}+p_{\hat{i}}\leq \frac{16}{3} \rbr{ 1-\frac{2\cost_{\hat{i}}}{p_{\hat{i}}}+p_{\hat{i}}}.
\end{align*}


Suppose we are given an optimal set $\optset$. Divide it into two sets based on the probability.
\begin{align*}
    \optset_1=\lbr{i\in\optset\setminus\{\hat{i}\}: p_i>\rbr{ 1-\frac{2\cost_{\hat{i}}}{p_{\hat{i}}}}};\qquad \optset_2=\lbr{i\in\optset\setminus\{\hat{i}\}: p_i\leq\rbr{ 1-\frac{2\cost_{\hat{i}}}{p_{\hat{i}}}}}.
\end{align*}

For the set $\optset_1$, by Lemma \ref{lem:asym-upper-bound}, there are at most $\sfrac{16}{3}$ tasks in $\optset_1$. By picking the most valuable task among $\optset$, the set $\effortset'^j$ achieve a $\sfrac{16}{3}$-approximation to the value of $\optset_1$.



For the set $\optset_2$, we take the knapsack solution with a budget reduced by $\frac{16}{3}$ factor. 
By enumerating over the budget-pivotal task $\hat{i}$, the recommendation set in \Cref{fig:threshold scoring mechanism} provides a $\sfrac{32}{3}$-approximation to the value of $\optset_2$. 

\end{enumerate}


Combining the above two cases, we have
\begin{align*}
\rbr{\frac{16}{3}+\frac{32}{3}}\val(\effortset) \geq \val(\optset_1)+\val(\optset_2) = v(\optset),
\end{align*}
which implies the recommendation set $\effortset$ is a $16$-approximation to the value of $\optset$.

\item $Y_3 =\lbr{i: \frac{16}{15} < \frac{p_i}{2\cost_i} \leq 11}$.
In this case, for any set $\effortset\subseteq Y_3$ that can be incentivized, and any $i^*\in\effortset$, 
we have
\begin{align*}
\sum_{i\in \effortset\backslash\lbr{i^*}} p_i \leq 
\sum_{i\in \effortset\backslash\lbr{i^*}} 22 \cost_i
\leq 22 \leq 352 \rbr{1-\frac{2\cost_{i^*}}{p_{i^*}}}
\end{align*}
where the last inequality holds since 
$\frac{2\cost_{i^*}}{p_{i^*}}\leq\frac{15}{16}$.
By the same argument as case 2, the threshold mechanism is a $1056$-approximation to the optimal in the knapsack scoring problem when the set of available tasks is $Y_3$.
\end{enumerate}

Combining all cases, for additive valuations, the maximum between truncated separate scoring rule 
and threshold scoring rule is a
$1091$-approximation to the optimal value \opt,
and the parameters can be computed in polynomial time.
Finally, for submodular valuation, 
the only difference is that the greedy solution we adopted for finding the set of incentivizable tasks 
loses an additional approximation factor of $\sfrac{e}{(e-1)}$ in valuations \citep*{sviridenko2004note}. 
Note that this additional factor can be save if we don't require computational efficiency 
and brute force search for the optimal set that can be incentivized given our proposed scoring rule. 
\end{proof}

\section{Sequential Effort}
\label{sec:search}
In this section, we show that the value approximation results for the static effort model can be generalized 
when the effort choice is made sequentially by applying the same family of scoring rules.

In the sequential search model, the agent makes the effort choice on tasks sequentially with the order of his choice.
Our designed scoring rule is robust against the strategy the agent adopts for exerting efforts on the recommendation set
as long as the strategy is not obviously dominated.

\begin{theorem}\label{thm:sequential-approx}
The better of 
a truncated separate scoring rule
and a threshold scoring rule 
is a $561$-approximation to the optimal value of the knapsack scoring problem (\opt)
when the agent does not adopt obliviously dominated strategies.
Moreover, for additive values, the parameters of such mechanism can be computed in polynomial time, 
and for submodular values, 
there is a polynomial time algorithm for computing the parameters that loses an additional multiplicative factor of $\sfrac{e}{(e-1)}$ in approximation ratio. 
\end{theorem}


\begin{proof}
Again, we first prove the theorem for additive valuations.
Similarly as \Cref{thm:approximation}, we divide the tasks into two sets $X, Y$ based on the ratio $\sfrac{p_i}{2\cost_i}$ as follows
$$
X = \lbr{i: \frac{p_i}{2\cost_i} > 11}; 
\qquad Y =\lbr{i: 1 \leq \frac{p_i}{2\cost_i} \leq 11}.
$$

On set $X$, the mechanism in \Cref{fig:truncated scoring mechanism} with budget $1$ achieves a $\frac{99}{8}$-approximation. Let the last assignment completed be $i$. 
By the same proof of \Cref{thm:budget inflation},  
for any task $i\in \effortset$, the probability that the scoring rule runs out of budget before the agent exerting effort on task $i$ can be bounded by $\frac{8}{9}$. 
Hence, when adopting strategies that are not obviously dominated, 
with ex ante probability at least $\frac{8}{9}$, 
the agent will stop after finishing all the tasks in the recommendation set. 
The same mechanism looses another $\frac{8}{9}$ factor in the approximation ratio.

On set $Y$, we divide the tasks into two sets by the probability $p_i$ of knowing the truth.

\begin{align*}
Y_1 = \lbr{i: p_i\geq 0.1}; \quad
Y_2 =\lbr{i: p_i< 0.1}.
\end{align*}

On set $Y_1$, it is sufficient to pick the highest-value task and post the threshold scoring rule. By the probability-cost ratio $\frac{p_i}{\cost_i}\leq 22$, each task has $\cost_i\geq \frac{1}{220}$. At most $440$ tasks can be incentivized in $Y_1$. Hence a $440$-approximation on $Y_1$.

\begin{figure}[t]
    \centering
    \fbox{
    \parbox{0.96\textwidth}{

        \textbf{Recommendation set $\effortset$} for threshold mechanism, with sequentially responding agent
        
        Input: ground set $\groundset$
        
        Output: set $\effortset$
        
        Greedily add tasks from $\groundset$ to $\effortset$, by value-probability ratio $\frac{\val_i}{p_i}$ with a budget $0.55$ on the total probabilities $\sum_{i\in \effortset}p_i$ of knowing the truth.
    }
    }
    \caption{Procedure for identifying approximately optimal recommendation set with sequentially responding agent.}
    \label{fig:threshold scoring mechanism sequential}
\end{figure}


        


On set $Y_2$, we use the scoring mechanism in \Cref{fig:threshold scoring mechanism sequential}.
We show this mechanism achieves a $109$-approximation, by showing when the adopted strategy is not obviously dominated: 
(1) with probability at least $0.45$, the agent completes all the tasks in the recommendation set; 
and 
(2) the total value in the set  is a $109$-approximation.

\begin{itemize}
    \item The agent completes tasks in recommendation set with probability at least $0.45$. 
    By union bound, the probability that the agent gets any informative signal is $1-\prod_i(1-p_i)\leq \sum_i p_i\leq 0.55$. For any order of completing the task, the agent gets no informative signal with probability at least $0.45$. The marginal gain of doing one more task is always positive, so the agent will finish the recommendation set with probability at least $0.45$. 
    \item The total value in the set is a $49$-approximation to the optimal.  All tasks in $Y_2$ has $p_i<0.1$, so by setting the budget at $0.55$, the total probabilities in $\effortset$ is at least the optimal knapsack value with budget $0.45$ on total probabilities. Since the probability-cost ratio $\frac{p_i}{c_i}\leq 22$, there is a budget on the total probabilities in any set that can be incentivized: $\sum_{i}p_i\leq 22$. Hence a $49$-approximation.
\end{itemize}
Combining the claims above, the better of the truncated scoring mechanism and the threshold scoring mechanism achieves a $561$-approximation when the agent is responding sequentially. 
\end{proof}

\bibliography{ref.bib}
\newpage
\appendix
\section{Probability Tools}
\label{apx:prob}

\begin{lemma}[\citealp{hoeffding1963probability}]
\label{lem:hoeffding}
Suppose $X_1, ..., X_n$ are independent random variables such that $X_i\in[a_i,b_i]$. Let $X=\sum_i X_i$. 
For any $\delta > 0$,
\begin{align*}
\prob{X - \expect{X} \geq \delta}
\leq \exp\left(-\frac{2\delta^2}{\sum_i (b_i-a_i)^2}\right).
\end{align*}
\end{lemma}

\begin{lemma}[\citealp{bernstein1927the}]
\label{lem:bernstein}
Suppose $X_1, ..., X_n$ are independent zero-mean random variables such that $|X_i|\leq M$. 
Let $X=\sum_i X_i$.
For any $\delta > 0$,
\begin{align*}
\prob{\abs{X} \geq \delta}
\leq 2\exp\left(-\frac{\frac{1}{2}\delta^2}{\sum_i \expect{X^2_i}+\frac{M}{3}}\right).
\end{align*}
\end{lemma}


\begin{lemma}[\citealp{pinsker1960information}]
\label{lem:pinsker}
If $P$ and $Q$ are two probability distributions on a measurable space $(X, \Sigma_{X})$, then for any measurable event $\event \in \Sigma_{X}$, it holds that
\[
\left| P(\event) - Q(\event) \right| \leq \sqrt{\frac{1}{2} \mathrm{KL}(P \| Q)},
\]
where 
\[ 
\mathrm{KL}(P \| Q) = \int_X \left(\ln \frac{\mathrm d P}{\mathrm d Q}\right) \mathrm d P
\]
is the Kullback--Leibler divergence.
\end{lemma}

\section{Properness for Belief Elicitation}
\label{apx:proper belief}
The main idea of converting any scoring rule that is potentially not proper for some beliefs to a scoring rule that is proper for all beliefs 
is to apply the taxation principle and let the agent chooses his best option given the original scoring rule. 
\begin{definition}\label{def:proper belief}
A scoring rule
$\score: \Delta(\outcomes)\times\outcomes\to [0,1]$ is \emph{proper for belief elicitation} if 
for any $\mu, \mu'\in \Delta(\outcomes)$, 
\begin{align*}
    \expect[\omega\sim\mu]{\score(\mu,\outcome)} \geq \expect[\omega\sim\mu]{\score(\mu',\outcome)}.
\end{align*}
\end{definition}
\begin{claim}\label{claim:proper belief}
For any proper scoring rule $\score$, there exists another scoring rule $\hat{\score}$ that is proper for belief elicitation such that 
\begin{align*}
\score(\signal,\outcome) = \hat{\score}(\mu(\signal),\outcome)
\end{align*}
for any $\signal\in\signals$ and $\outcome\in\outcomes$
where $\mu(\signal)$ is the posterior belief of the agent when receiving signal $\signal$.
\end{claim}
\begin{proof}
Consider the following scoring rule for belief elicitation:
\begin{align}\label{eq:best response}
    \hat{\score}(\mu, \outcome)=\score(\signal^*, \outcome), \text{ where }\signal^*\in\argmax_{\signal}\expect[\outcome\sim\mu]{\score(\sigma, \outcome)}.
\end{align}
Next, we will show that
\begin{enumerate}
\item $\hat{\score}$ is proper for belief elicitation. Let $\signal^*(\mu)=\argmax_{\signal}\expect[\outcome\sim\mu]{\score(\sigma, \outcome)}$ be the best responding signal when the agent has to choose a signal to report. 
For any belief $\mu$ and $\mu'$, we have
\begin{align*}
    \expect[\outcome\sim\mu]{\score(\mu', \outcome)}
    =\expect[\outcome\sim \mu]{\score(\signal^*(\mu'), \outcome)}
    \leq \expect[\outcome\sim\mu]{\score(\signal^*(\mu), \outcome)}
    =\expect[\outcome\sim\mu]{\score(\mu, \outcome)}
\end{align*} 
which implies that $\hat{\score}$ is proper for belief elicitation. 

\item $\hat{\score}$ is an extension of $\score$, i.e.\ $\score(\signal,\outcome) = \hat{\score}(\mu(\signal),\outcome)$. 
This follows directly from the properness of the original scoring rule $\score$. \qedhere
\end{enumerate}
\end{proof}

Note that given an arbitrary scoring rule $\score$, 
computing the best response strategy $\signal^*$ given his belief $\mu$ as in Equation \eqref{eq:best response} may be NP-hard. 
Therefore, even though such proper scoring rule exists, 
it might be challenging to provide its exact form in polynomial time given our designed scoring rules. 
Fortunately, for our purpose of incentivizing effort, 
we can adopt a similar solution concept in our sequential effort model (c.f., \cref{def:obvious dominated} and \cref{sec:search}) by allowing the agent to approximately best response to the scoring rule. 
More specifically, given any scoring rule $\score$ for eliciting the signals, 
the principal can offer this original scoring rule $\score$ to the agent, 
ask the agent to report his belief, and let the agent choose the best possible signal he can find in polynomial time as input to the scoring rule $\score$ for computing his score based on his belief.
This protocol disentangles the incentives between reporting beliefs and maximizing the expected score, 
and hence the agent has no incentive to misreport his true belief. 
Moreover, since the scoring rule is proper for all signals in the support, 
for any belief induced by those signals, the agent's best response is to simply report those signals truthfully. 
For any belief that cannot be induced by those signals, 
the agent can adopt any polynomial time algorithm for finding an approximately optimal solution. 
However, as those events happen with probability measure 0, 
it would not affect the agent's incentives for exerting effort in our model, 
and all of our results extend naturally.

\section{Missing Proofs and Constructions}
\subsection{Alternative Formulation of Threshold Scoring Rules}
\label{apx:exp proper for belief}
Here we present an alternative formulation of the threshold scoring rule in the special case of threshold $1$
given outcome space $\outcomes=\{0, 1\}^n$.

\begin{definition}[\citealp*{HLSW-20}]\label{def:optimal mos}
Consider the $n$-dimensional outcome space $\outcomes=\{0, 1\}^n$.
Given single-dimensional scoring rules
\begin{align*}
\score_i(\mu_i, \outcome_i)=\begin{cases}
1& \mu_i\leq \sfrac{1}{2} \text{ and } \outcome_i=0, \text{ or } \mu_i> \sfrac{1}{2} \text{ and } \outcome_i=1,\\
0& \text{otherwise}.
\end{cases}
\end{align*}
The canonical \mos scoring rule $\score$ 
is defined as 
\begin{align*}
    \score(\mu, \outcome)=\score_i(\mu_i, \outcome_i), \text{ where }i=\argmax\expect[\outcome_i\sim\mu_i]{\score_i(\mu_i, \outcome_i)},
\end{align*}
\end{definition}

By \citet*{HLSW-20}, the canonical \mos scoring rule is proper for belief elicitation. 
Moreover, it is easy to verify that it coincides with the threshold scoring rule with threshold $1$ given any belief of the agent.

\subsection{Proof of \cref{thm:budget inflation}}
\label{apx:proof budget inflation}

We show that the mechanism in \Cref{fig:truncated scoring mechanism}  is incentive compatible,
by first showing that scoring rule $S$ is proper, 
and then showing that $\effortset$ is the agent's best effort choice. 

\paragraph{Proper.} 
For each task $i$, conditional on receiving signal $\signal_i\neq\bot$,
the score $\score_i(\signal_i,\outcome_i)$ first order stochastically dominates $\score_i(\signal'_i,\outcome_i)$ for any $\signal'_i$.
Thus, the agent has incentives to truthfully report the signal $\signal_i$ if $\signal_i\neq\bot$.

We then show that the agent has no incentives to misreport on tasks with uninformative signal $\signal_i = \bot$ by contradiction.
Suppose that the agent has incentives to misreport given signal $\bot$ on some tasks. 
We partition the tasks into three sets. 
Let $Z_0$ be the set of tasks $i$ such that $\signal_i\neq \bot$, 
$Z_1$ be the set of tasks $i$ such that $\signal_i= \bot$ and where the agent truthfully reports the signal,
and $Z_2$ be the set of tasks $i$ such that $\signal_i= \bot$ and where the agent misreports the signal. 
First note that if 
$\sum_{i\in Z_0} 2\cdot\sfrac{9c_i}{8p_i} \geq 11$,
then by truthful reporting the signals the agent can secure a deterministic score $11$, 
which is the maximum possible score. 
Hence the agent has no incentive to misreport in this case.

Next we focus on the case when $\sum_{i\in Z_0} \sfrac{9c_i}{8p_i} < \sfrac{11}{2}$.
Let $\eta_i$ be a Bernoulli random variable with probability $\sfrac{1}{2}$ drawn independently for each task $i\in Z_2$. We use $\eta_i$ to indicate whether the agent guesses correctly on the task $i \in Z_2$. Let 
\begin{align*}
s = \sum_{i\in Z_0} \frac{9c_i}{8p_i} + \sum_{i\in Z_2} \frac{9c_i}{4p_i}\left(\eta_i - \frac{1}{2}\right) + \frac{11}{2}.
\end{align*}
Note that $s$ is the random variable corresponding to the score without truncation by the interval $[0,11]$.
Consider an alternative setting where the score is truncated by the interval $[\sum_{i\in Z_0} \sfrac{9c_i}{4p_i}, 11]$. 
Since the distribution of $s$ is symmetric with respect to the mean $\rbr{\sum_{i\in Z_0} \sfrac{9c_i}{4p_i} + 11}/2$, the score distribution under the truncation by $[\sum_{i\in Z_0} \sfrac{9c_i}{4p_i}, 11]$ is also symmetric with respect to the mean.
Thus, the utility of the agent for misreporting in this alternative setting is exactly the same as the utility for truthful reporting, $\rbr{\sum_{i\in Z_0} \sfrac{9c_i}{4p_i} + 11}/2$. 
Since $\sum_{i\in Z_0} \sfrac{9c_i}{4p_i} > 0$, the utility of the agent for misreporting with truncation by $[0,11]$ is strictly less than the utility for misreporting with truncation by $[\sum_{i\in Z_0} \sfrac{9c_i}{4p_i}, 11]$.  
Therefore, the agent will not have an incentive to misreport in the original setting when the lower truncation is $0$.

\paragraph{Effort Set.}
We prove that the agent's optimal choice is to exert effort in tasks $\effortset$.
First note that we set the score to be zero for $i\not\in\effortset$. 
This immediately implies that the agent will not exert effort on task $i\not\in\effortset$. 
Fix the agent's effort choice in $\effortset$.
Suppose there exists a task $\hat{i}\in \effortset$ such that the agent's effort on task $\hat{i}$ is $0$.
Let~$\hat{\event}_{\hat{i}}$ be the event that $-d+\sum_{i\in\effortset \backslash \{\hat{i}\}} \score_{i}(\signal_{i},\outcome_{i}) \in [0,11-\sfrac{9c_{\hat{i}}}{8p_{\hat{i}}}]$. 
Let $\hat{Z}\subseteq\effortset$ be the set on which the agent exerts effort.
Therefore, 
\begin{align*}
\prob{\hat{\event}_{\hat{i}}} &= 1-
\prob{-d+\sum_{i\in\effortset \backslash \{\hat{i}\}} \score_{i}(\signal_{i},\outcome_{i}) > 11-\frac{9c_{\hat{i}}}{8p_{\hat{i}}}}\\
&= 1-
\prob{\sum_{i\in\hat{Z}} \indicate{\signal_i \neq \bot} \cdot \frac{9c_i}{8p_i} > \frac{11}{2}-\frac{9c_{\hat{i}}}{4p_{\hat{i}}}}\\
&\geq 1-\exp\rbr{-\frac{\frac{1}{2}\rbr{\frac{11}{2}-\frac{9c_{\hat{i}}}{4p_{\hat{i}}} 
- \sum_{i\in\hat{Z}}\frac{9 c_i}{8}}^2}
{\frac{1}{4}\sum_{i\in\hat{Z}} \frac{1}{p_i}\cdot \frac{9c_i}{4}^2
+ \frac{1}{6}\max_{i\in\hat{Z}} \frac{9c_i}{4p_i}}} \\
&\geq 1-\exp\rbr{-\frac{\rbr{11 - \frac{45}{8}}^2}{6\cdot \frac{9}{8}^2+\frac{3}{2}}} \geq \frac{8}{9},
\end{align*}
where the first inequality holds by applying Bernstein's inequality (\cref{lem:bernstein}).
The second inequality holds since (1) $\sum_{i\in\hat{Z}} \frac{9c_i}{4} = \frac{9}{4} \sum_{i\in\hat{Z}} c_i \leq \frac{27}{8}$; 
(2) $\sum_{i\in\hat{Z}} \sum_{i\in\hat{Z}} \frac{1}{p_i}\cdot \frac{9c_i}{4}^2 \leq \sum_{i\in\hat{Z}} 2(9/8)^2 c_i \leq 3(9/8)^2$; 
and (3) $\max_{i\in\hat{Z}} \sfrac{9c_i}{4p_i} \leq 9/8$.
Hence, by exerting effort on task $\hat{i}$, 
the score of the agent increases by at least $\prob{\hat{\event}_{\hat{i}}} \cdot \sfrac{9 c_{\hat{i}}}{8}  \geq c_{\hat{i}}$, which provides
a contradiction. 


For submodular values, we lose a $\sfrac{e}{e-1}$ factor in the value approximation ratio by computing the recommendation set $\effortset$ in polynomial time (\cref{lem:submodular greedy}). Without computation constraints, we have a scoring rule that achieves the theoretical bound. 

\section{General Information Structure}
\label{sec:general info}

In this section, we consider the problem of incentivizing effort with general information structures
and illustrate the intrinsic challenges for generalizing our results to general information structures.
Here, when the agent exerts effort, 
instead of assuming that the he observes the true state $\outcome_i$ with probability $p_i$ as in previous sections,
the agent receives a signal $\sigma_i \in \signals$ given by a signal structure that 
induces a distribution 
$\distoverposterior_i$ over posterior $\posterior_i\in \Delta(\outcomes)$. We show that the optimal value of the knapsack scoring problem can differ a lot under two different information structures
even if the optimal scoring rules for the single task problems are the same given those two information structures. 
Therefore, new ideas for designing approximately optimal scoring rules are required for general information structures. 

First, the following lemma characterizes whether a single task can be incentivized by an incentive compatible mechanism
under general information structure environments. 
\begin{lemma}[\citealp{HLSW-20}]
\label{lem:single opt-general}
For the knapsack scoring problem with general information structures, the agent can be incentivized to exert effort on a single task $\effortset = \{i\}$ with budget $1$ 
if and only if 
$$
\expect[\posterior_i\sim\distoverposterior_i]{\abs{\posterior_i-\prior}}\geq c_i,
$$ 
where $\abs{\posterior_i-\prior}$ is the difference of the mean between the posterior and the prior.
\end{lemma}

When there are multiple tasks, 
a crucial statistic that affects the set of the incentivizable tasks 
is the expected \KL-divergence between the prior and the posterior.
Specifically, let 
\begin{align*}
\divergence_i \triangleq \expect[\posterior_i\sim\distoverposterior_i]{\KL(\prior\| \posterior_i)}
\end{align*}
where $\KL(\prior \| \posterior_i) = \sum_{\outcome\in\outcomes} \prior(\outcome)\cdot 
\ln \frac{\prior(\outcome)}{\posterior_i(\outcome)}$
is the \KL-divergence between the prior $\prior$ and the posterior~$\posterior_i$.
This distance measures how easy for the agent to mimic the signal distributions without exerting effort. 
The following lemma provides an upper bound on the set of incentivizable tasks given asymmetric and general information structures.

\begin{lemma}\label{lem:general info upper}
For the knapsack scoring problem with general information structures, for any set $\effortset^*$ such that there exists an incentive compatible mechanism 
where the agent's optimal effort choice is~$\effortset^*$, 
we have 
\begin{align*}
\sum_{i\in\optset}\cost_i \leq 
\sqrt{\frac{1}{2} \sum_{i\in\optset} \divergence_i}.
\end{align*}
\end{lemma}
\begin{proof}
Note that given any proper scoring rule $\score$, 
one feasible choice of the agent is to exert no effort, 
simulate the posterior distribution on set $\optset$, 
and report the simulated posterior to the principal. 
Let $P$ be the distribution over the profile of reports, and states for all tasks in $\optset$
given the simulations on $\effortset^*$. Let $Q$ be such distribution
when the agent exerts effort on all tasks in~$\optset$
and get the true informative signals. 
It is easy to verify that the \KL-divergence between
$P$ and $Q$ is $\sum_{i\in\optset} \divergence_i$. 
Let $\event$ be the event such that the profile of reports and states does not coincide 
given the true posterior generating process and the simulated reports.
Then we have 
\begin{align*}
&\expect[\sigma\sim\effortset]{\expect[\omega\sim\signal]{\score(\signal,\outcome)}} 
- \expect[\sigma\sim\emptyset]{\expect[\omega\sim\signal]{\score(\signal,\outcome)}} 
\leq \expect[Q]{\score(\signal,\outcome)} 
- \expect[P]{\score(\signal,\outcome)} \\
& \leq \abs{\prob[P]{\event} - \prob[Q]{\event}} 
\leq \sqrt{\frac{1}{2}\KL(P \| Q)}
= \sqrt{\frac{1}{2} \sum_{i\in\optset} \divergence_i}
\end{align*}
where the second inequality holds since the payment of the principal is at most 1, 
and the third inequality holds by Pinsker's inequality (\cref{lem:pinsker}).
\end{proof}

Next we show that given two different information structures such that 
the design of the optimal scoring rule for both cases are the same 
in the single task problem, 
the set of incentivizable tasks may differ a lot when there are multiple tasks. 

Specifically, consider the symmetric environment with identical information structures and costs~$\cost$ for all tasks, 
\cref{lem:general info upper} implies that 
$\abs{\optset} \leq \frac{\divergence}{2\cost^2}$. 
Fixing $p>0$, consider the following two information structures when the agent exerts effort on any single task:
\begin{itemize}
    \item the agent receives an informative signal $\signal = \outcome$ with probability $p$,
    and receives an uninformative signal $\signal=\bot$ regardless of the realized state with probability $1-p$; 
    \item the agent receives an informative signal that induces posterior $\posterior = \frac{1+p}{2}$ and $\frac{1-p}{2}$ with probability $\frac{1}{2}$ each.
\end{itemize}

Given both information structures above, in the single task problem, by \cref{lem:single opt-general}, we know that the agent can be incentivized to exert effort on the single task
if and only if the cost of effort is at most $\sfrac{p}{2}$.

For the multi-task problem, suppose that the cost of effort on a single task is $\cost = \frac{p}{4}$. 
Given the first information structure, it is easy to show that the optimal scoring rule can incentivize the agent to exert effort on $O(\frac{1}{c})$ tasks.
By \cref{thm:approximation},
the agent can be incentivized to exert effort on $O(\frac{1}{c})$ tasks by the threshold scoring rule. 
In contrast, given the second information structure, 
we have that $\divergence = O(p^2)$
and hence by \cref{lem:general info upper}, 
the size of the incentivizable tasks is at most 
$\frac{\divergence}{2\cost^2} = O(1)$. 
The gap on the size of the incentivizable tasks between two different information structures are unbounded when $p$ and $c$ are sufficiently small.

The above observation indicates that the design of the (approximately) optimal scoring rules depends on the fine details of different information structures
even if they have the same performance on the single task problem. 
Thus it is unlikely to directly generalize our results for the special case to general information structures,
or derive a unified approach for reducing the multi-task knapsack scoring problems to single-task ones.
It is an interesting open question to identify tight upper bounds of the optimal solution for the knapsack scoring problem with general information structures, 
and design approximately optimal scoring rules to approximate the upper bound. 


\end{document}